\documentclass{amsart}
\usepackage{amssymb, amsmath}%%% 
\usepackage{graphics}%%% 
\usepackage[dvips]{graphicx}%%% 
\usepackage{eepic}%%% 
\usepackage{epic}%%% 
\usepackage{array}
\usepackage{bm}
\usepackage{exscale}
\usepackage{latexsym}
\newtheorem{thm}{Theorem}[section]

\newtheorem{prp}[thm]{Proposition}

\newtheorem{dfn}[thm]{Definition}
\def\C{{\mathbb C}}
\def\Z{{\mathbb Z}}
\def\Q{{\mathbb Q}}
\def\F{{\mathbb F}}

\def\P{{\mathbb P}}

\def\II{${}_{\mbox{\scriptsize{II} }}$}
\def\I{${}_{\mbox{\scriptsize{I} }}$}

\title{Several results on the almost good reduction property to the $q$-discrete Painlev\'{e} equations over the field of $p$-adic numbers}

\author{Masataka Kanki$^1$ \\
\small $^1$ Faculty of Engineering Science, Kansai University, 3-3-35 Yamate, Osaka 564-8680, Japan}

\date{}

\begin{document}

\maketitle

\begin{abstract}
We report on the recent results in the theory of discrete integrable equations reduced modulo a prime number $p$ based on the previous results \cite{Kanki}. This article mainly concerns the $q$-discrete analogues of the Painlev\'{e} equations. We first give a quick review of the idea of `almost good reduction', which is a property that most integrable systems over a finite field $\mathbb{F}_p$ possess when reduced from those defined over the field of $p$-adic numbers $\mathbb{Q}_p$.
This idea has been introduced in our previous papers and the discrete and $q$-discrete Painlev\'{e} II equations have been proved to have `almost good reduction' \cite{KMTT, JNMP}, which can be interpreted as an arithmetic analogue of the singularity confinement test.
In this article we prove that this property holds for several other $q$-discrete Painlev\'{e} equations. 

The theory of Diophantine integrability \cite{Halburd} is reviewed and its relation to the integrability detection by the algebraic entropy is discussed.
\end{abstract}

\section{Introduction}
%%%%%%%%%%%%%%%%%%%%%%%%%%%%%%%%%%%%%%%%%%%%%%%%%%
The purpose of our research is to define and investigate the discrete integrable equations over finite fields. We wish to study the meaning of integrability over finite fields.

In the case of linear discrete equations, we can well-define the equations over finite fields just by changing the field on which the equations are defined to finite fields. However, in the case of nonlinear equations, since the systems are usually formulated by rational functions, the division by $0$ mod $p$ and indeterminacies such as $0/0$ and $\infty\pm \infty$ frequently appear.
These points makes it difficult for us to well-define the equations over finite fields. Thus there has been few studies on the nonlinear discrete integrable equations defined over finite fields.

There are mainly three approaches to overcome this difficulty.
The first one is to study the equation in the bilinear form.
Bilinear form of the discrete KP and KdV equations have been treated over the finite field \cite{Bialecki,DBK,BD}.
The second one is to restrict the domain of definition of the system so that the indeterminacies do not appear: discrete Toda equation over finite fields and the graphical structures of the solutions have been obtained \cite{ytakahashi}.
The third one is to extend the space of initial conditions to make the mapping well-defined.

We investigate the third approach in this paper and try two different schemes.
The first scheme is to apply the theory of the space of initial conditions developed by K. Okamoto \cite{Okamoto}  and H. Sakai  \cite{Sakai}.
According to the Sakai theory, the domain of definition of the discrete Painlev\'{e} equation becomes a birational surface as we extend the domain by blowing-up at each singular point.
We have shown in our previous work \cite{JNMP}, that this procedure is still valid if applied to the domain of definition $\mathbb{F}_{p^m}\times \mathbb{F}_{p^m}$. 

The second scheme of extension is to define the equations over the field of $p$-adic numbers and then reduce them to the finite fields.
For example, if we try to define the discrete Painlev\'{e} equation over the field $\F_p$,
the initial value space is a finite set $\F_p \times \F_p$. Since the system consists of transitions between just $p^2$ points, it is not clear how  we can formulate the integrability of the system from the integrability of the original system defined over $\mathbb{C}$. To resolve this problem we consider a pair of fields $(\Q_p,\ \F_p)$. We can say that the system over $\F_p$ is `integrable' if it is integrable over $\Q_p$ and its reduction to the finite field $\F_p$ has so-called `almost good reduction' property.
This approach is closely related to the theory of arithmetic dynamical systems, which concerns the dynamics over arithmetic sets such as $\mathbb{Z}$ or $\mathbb{Q}$ or a non-Archimedean field that is of number theoretic interest \cite{Silverman}.
In arithmetic dynamics, the polynomial or rational mappings reduced modulo a prime numbers give significant information on the original systems. 
The mapping is said to have good reduction if, roughly speaking, the reduction commutes with the mapping itself \cite{Silverman}.
The QRT mappings \cite{QRT} and several birational mappings over finite fields have been investigated in terms of integrability by choosing the parameter values so that indeterminate points are avoided
\cite{Roberts,Roberts2,Roberts3}.
We have proved that, although the discrete and $q$-discrete Painlev\'{e} equations do not have a good reduction modulo a prime,
they do have an \textit{almost good reduction}, which is a generalized notion of good reduction \cite{KMTT,JNMP}.
In this paper, we first review some of the previous results, and then apply the method to several of the $q$-discrete Painlev\'{e} equations previously not published, and prove the almost good reduction properties of the equations.
The results further strengthen our idea that the almost good reduction can be used as an integrability detector.
%We then prove that the Hietarinta-Viallet equation has almost good reduction, and conclude that the almost good reduction is an arithmetic analogue of the singularity confinement test, that is to say, a kind of `$p$-adic singularity confinement'.
Finally we discuss the relation of our $p$-adic approach to the theory of Diophantine integrability \cite{Halburd}.

%%%%%

\section{The almost good reduction property to the discrete equations (review)}
The discrete Painlev\'{e} equations are non-autonomous, integrable mappings which tend to some continuous Painlev\'{e} equations for appropriate choices of the continuous limit \cite{RGH}.
They are non-autonomous and nonlinear second order ordinary difference equations with several parameters.
When they are defined over a finite field, the dependent variable takes only a finite number of values and their time evolution will attain an indeterminate state in many cases for generic values of the parameters and initial conditions.

\subsection{Almost good reduction}
%
%%%%%%%%%%%%%%%%%%%%%%%%%%%%%%%%%%%%%%%%%%%%%%%%%%%%%%%%%
Let $p$ be a prime number. Each rational number $x \in \Q$ ($x \ne 0$) can be written uniquely as $x=p^{v_p(x)} \dfrac{u}{v}$ where $v_p(x), u, v \in \Z$ and $u$ and $v$ are coprime integers neither of which is divisible by $p$. The integer $v_p(x)$ is called the $p$-adic valuation of $x$.
The $p$-adic norm $|x|_p$ is defined as $|x|_p=p^{-v_p(x)}$. ($|0|_p=0$.)
The local field $\Q_p$ is a completion of $\Q$ with respect to the $p$-adic norm. 
It is called the field of $p$-adic numbers and its subring $\Z_p:=\{x\in \Q_p | \ |x|_p \le 1(\leftrightarrow v_p(x)\ge 0)\}$ is called the ring of $p$-adic integers \cite{Murty}. 
The $p$-adic norm satisfies a non-Archimedean (ultra-metric) triangle inequality 
$|x+y|_p \le \max[|x|_p,|y|_p ]$.
The subset $\mathfrak{p}=p\Z_p=\left\{x \in \Z_p |\ v_p(x) \ge 1 \right\}$ is the unique maximal ideal of $\Z_p$.
We define the reduction of $x$ modulo $\mathfrak{p}$ by the natural projection
\[
\pi: \Z_p \ni x \mapsto \pi(x) \in \Z_p/\mathfrak{p} \cong \F_p,
\]
and write $\tilde{x}:=\pi(x)$.
Note that the reduction map $\pi$ is a ring homomorphism.
The map $\pi$ is extended to general elements of $\Q_p$ as:
\[
\Q_p \ni x\mapsto
\left\{\ 
\begin{array}{ll}
\pi(x)\in\mathbb{F}_p & (x\in\mathbb{Z}_p)\\
\infty & (x\in\mathbb{Q}_p\setminus \mathbb{Z}_p)
\end{array}
\right. \in \P\F_p,
\]
which is no longer homomorphic.
For a rational system $\phi(x,y)\in(\mathbb{Z}_p(x,y))^2$: $\mathcal{D} \subseteq \Z_p^2 \to \Z_p^2$, defined on some domain $\mathcal{D}$, the reduced system $\pi(\phi)(x,y)=\tilde{\phi}(x,y)\in(\mathbb{F}_p(x,y))^2$ is defined as the system whose coefficients are all reduced.
The rational system $\phi$ is said to have a \textit{good reduction} (modulo $\mathfrak{p}$ on the domain $\mathcal{D}$) if we have $\widetilde{\phi(x,y)}=\tilde{\phi}(\tilde{x},\tilde{y})$ for any $(x,y) \in \mathcal{D}$ \cite{Silverman}.
We have defined a generalised notion in our previous letter and have explained its usefulness;
\begin{dfn}[\cite{KMTT}]
A (non-autonomous) rational system $\phi_n$: $\Q_p^2 \to (\mathbb{P}\Q_p)^2$ $(n \in \Z)$ has an \textit{almost good reduction} modulo $\mathfrak{p}$ on the domain $\mathcal{D}^{(n)}\subseteq \Z_p^2\cap \phi_n^{-1}(\Q_p^2)$, if there
exists a positive integer $m_{\mbox{\rm \scriptsize p};n}$ for any $\mbox{\rm p}=(x,y) \in \mathcal{D}^{(n)}$ and time step $n$ such that
\begin{equation}
\widetilde{\phi_n^{m_{\mbox{\rm \tiny p};n}}(x,y)}=\widetilde{\phi_n^{m_{\mbox{\rm \tiny p};n}}}(\tilde{x},\tilde{y}),
\label{AGR}
\end{equation}
where $\phi_n^m :=\phi_{n+m-1} \circ \phi_{n+m-2} \circ \cdots \circ \phi_n$.
\end{dfn}
Let us first review some of the findings in \cite{KMTT} in order to understand the almost good reduction. Let us consider the mapping $\Psi_\gamma(x_n,y_n)=(x_{n+1},y_{n+1})$ defined as
\begin{equation}
\Psi_\gamma:\left\{
\begin{array}{cl}
x_{n+1}&=\dfrac{ax_n+1}{x_n^\gamma y_n},\\
y_{n+1}&=x_n,
\end{array}
\right.
\label{discretemap}
\end{equation} 
where $a \in \mathbb{Z}_p^{\times}$($\leftrightarrow v_p(a)=0$)  and $\gamma \in \Z_{\ge 0}$ are parameters. 
The map \eqref{discretemap} is known to be integrable if and only if $\gamma=0,1,2$.
Note that when $\gamma=0,1,2$, \eqref{discretemap} is a type of QRT mappings \cite{QRT} and also an autonomous version of the $q$-discrete Painlev\'{e} I equation and therefore is integrable.
\begin{prp}[\cite{KMTT}] \label{qp1auto}
The rational mapping \eqref{discretemap} with $a\in\mathbb{Z}_p^{\times}$ and $\gamma\in\mathbb{Z}_{\ge 0}$ has an almost good reduction modulo $\mathfrak{p}$ on the domain $\mathcal{D}$ if and only if $\gamma=0,1,2$.
Here $\mathcal{D}=\Z_p^2\cap \Psi_{\gamma}^{-1}(\Q_p^2)$.
If $\gamma\ge 1$ then $\mathcal{D}=\{(x,y) \in \Z^2_p \ |x \ne 0, y \ne 0\}$, if $\gamma=0$ then $\mathcal{D}=\{(x,y) \in \Z^2_p \ |y \ne 0\}$.
\label{PropQRT}
\end{prp}
Note that if the domain does not depend on $n$, we simply write $\mathcal{D}^{(n)}=\mathcal{D}$.
Proof is omitted here.

On the other hand, for $\gamma \ge 3$ and $\tilde{x}_n=0$, we easily find that
\[
{}^\forall k \in \Z_{\ge 0}, \;\; \widetilde{\Psi_\gamma^{k}(x_n,y_n)} \ne \widetilde{\Psi_\gamma^{k}}(\tilde{x}_n=0,\tilde{y}_n),
\]
since the order of $p$ diverges as we iterate the mapping from $x_n=p^k\cdot e$ where $k>0,\ e\in\mathbb{Z}_p^{\times}$.

\smallskip

In this proposition we omitted the case of $a=0$.
However we can also treat this case.
In the case $\gamma=2$ and $a=0$, for example,
if we take 
\[
f_{2k}:=x_{2k}x_{2k-1},\ f_{2k-1}:=(x_{2k-1}x_{2k-2})^{-1}
\]
\eqref{discretemap} turns into the trivial linear mapping $f_{n+1}=f_n$ which has apparently good reduction modulo $\mathfrak{p}$.
In \cite{KMTT}, we have proved that discrete Painlev\'{e} II equation, too, has almost good reduction property.
Note that having an almost good reduction is equivalent to the integrability of the equation in these examples.

\section{A $p$-adic analogue of the singularity confinement test}

The above approach is closely related to the singularity confinement method which is an effective test to judge the integrability of the given equations \cite{Grammaticosetal}.
In the proof of the proposition \ref{PropQRT} we have taken $x_n=e p^k$ and have shown that the limit \[\lim_{|e p^k|_p \to 0}(x_{n+m}, x_{n+m+1})\]
is well defined for some positive integer $m$.
Here $ep^k\ (k>0,\ |e|_p=1)$ is an alternative in $\Q_p$ for the infinitesimal parameter $\epsilon$ in the singularity confinement test in $\C$. Note that $p^k\ (k>0)$ is a `small' number in terms of the $p$-adic metric.
From this observation and propositions \ref{PropQRT}, we postulate that having almost good reduction in arithmetic mappings is similar to passing the singularity confinement test. 

\section{The $q$-discrete Painlev\'{e} equations modulo a prime}
%
%
%%%%%%%%%%%%%%%%%%%%%%%%%%%%%%%%%%%%%%%%%%%%%%%%
In this section we study the $q$-discrete analogues of the Painlev\'{e} equations and one of the chaotic equations.

\subsection{The $q$P\I equation}
One of the forms of the $q$P\I equations is as follows:
\begin{equation}
x_{n+1}x_{n-1}=\frac{aq^nx_n+b}{x_n^2}, \label{qp1eq}
\end{equation}
where $a$ and $b$ are parameters.
We rewrite \eqref{qp1eq} for our convenience as
\begin{equation}
\Phi_n: \left\{
\begin{array}{cl}
x_{n+1}&=\dfrac{aq^n x_n+b}{x_n^2 y_n},\\
y_{n+1}&=x_n.
\end{array}
\right.
\label{qP1}
\end{equation} 
We can prove the following proposition:
\begin{prp}
Suppose that $a, b, q$ are integers not divisible by $p$, then the mapping \eqref{qP1} has an almost good reduction   
modulo $\mathfrak{p}$ on the domain 
$\mathcal{D}:=\Z_p^2\cap \Phi_n^{-1}(\Q_p^2)=\{(x,y)\in \Z_p^2\ |x \ne 0, y\ne 0\}$.
\label{PropqP1}
\end{prp}
The proof is essentially the same as in the proposition \ref{qp1auto} and is omitted here.

\subsection{The $q$P\II, $q$P$_{\mbox{\scriptsize III}}$, $q$P$_{\mbox{\scriptsize IV}}$ equations (review)}
The $q$P\II equation is the following $q$-discrete equation:
\begin{equation}
(z(q\tau)z(\tau)+1)(z(\tau)z(q^{-1}\tau)+1)=\frac{a \tau^2 z(\tau)}{\tau-z(\tau)},
\label{qP2eq}
\end{equation}
where $a$ and $q$ are parameters \cite{Kajiwaraetal}.
It is also convenient to rewrite \eqref{qP2eq} as
\begin{equation}
\Phi_n: \left\{
\begin{array}{cl}
x_{n+1}&=\dfrac{a(q^n\tau_0)^2x_n-(q^n\tau_0-x_n)(1+x_ny_n)}{x_n(q^n\tau_0-x_n)(x_ny_n+1)},\\
y_{n+1}&=x_n,
\end{array}
\right.
\label{qP2}
\end{equation}
where $\tau=q^n\tau_0$.
Just like the $q$P\I equation, we can prove that $q$P\II has an almost good reduction:
\begin{prp}[\cite{JNMP}]
Suppose that $a, q, \tau_0$ are integers not divisible by $p$, then the mapping \eqref{qP2} has an almost good reduction   
modulo $\mathfrak{p}$ on the domain 
$\mathcal{D}^{(n)}:=\Z_p^2\cap \Phi_n^{-1}(\Q_p^2)=\{(x,y)\in \Z_p^2\ |x \ne 0, x \ne q^n\tau_0, xy+1 \ne 0\}$.
\label{PropqP2}
\end{prp}

\smallskip

Next, we prove that a type of $q$-discrete Painlev\'{e} III, IV equations have an almost good reduction on some appropriate domains. (We place some restrictions on the parameters for simplicity.)

The $q$-discrete analogue of the Painlev\'{e} III equation has the following form
\[
x_{n+1}x_{n-1}=\frac{ab(x_n-cq^n)(x_n-dq^n)}{(x_n-a)(x_n-b)},
\]
where $a,b,c,d$ are parameters \cite{RGH}.
It is convenient to rewrite it as the following dynamical system
\begin{equation}
\Phi_n: \left\{
\begin{array}{cl}
x_{n+1}&=\dfrac{ab(x_n-cq^n)(x_n-dq^n)}{y_n(x_n-a)(x_n-b)},\\
y_{n+1}&=x_n.
\end{array}
\right.
\label{qP3}
\end{equation}
\begin{prp}[\cite{Kanki}]
Suppose that $a,b,c,d,q$ are in $\{1,2,\cdots,p-1\}$ and that $a,b,c,d$ are distinct and we also suppose that $a+b\not \equiv (c+d)q^3$, then the mapping \eqref{qP3} has an almost good reduction   
modulo $\mathfrak{p}$ on the domain 
$\mathcal{D}:=\{(x,y)\in \Z_p^2\ |x\neq a,b,\,y\neq 0 \}$.
\label{PropqP3}
\end{prp}

\smallskip

One of the $q$-discrete analogues of the Painlev\'{e} IV equation has following form:
\[
(x_{n+1}x_n-1)(x_nx_{n-1}-1)=\frac{aq^{2n}(x_n^2+1)+bq^{2n}x_n}{cx_n+dq^n},
\]
where $a,b,c,d$ are parameters \cite{RGH, RG}.
It can be rewritten as  follows:
\begin{equation}
\Phi_n: \left\{
\begin{array}{cl}
x_{n+1}&=\dfrac{\tau^2(ax_n^2+bx_n+a)+(x_ny_n-1)(x_n+\tau)}{x_n(x_ny_n-1)(x_n+\tau)},\\
y_{n+1}&=x_n,
\end{array}
\right.
\label{qP4}
\end{equation}
where $\tau=q^n\tau_0$.
Here we took $\tau_0=d/c$ and redefined $a,b$ as $ac/d^2 \to a$ and $bc/d^2 \to b$.
\begin{prp}[\cite{Kanki}]
Suppose that $|a|_p=|b|_p=|q|_p=|\tau_0|_p=1$, then the mapping \eqref{qP4} has an almost good reduction   
modulo $\mathfrak{p}$ on the domain 
\[
\mathcal{D}^{(n)}:=\{(x,y)\in \Z_p^2\ |x\neq 0,\ xy\neq 1\ x\neq -q^n\tau_0\},
\]
on condition that $aq^2\tau_0\neq 1,\ aq^4\tau_0\neq 1$.
\label{PropqP4}
\end{prp}

\subsection{The $q$P$_{\mbox{\scriptsize V}}$ equation}
The $q$-discrete Painlev\'{e} V equation has a following form:
\[
(x_{n+1}x_n-1)(x_nx_{n-1}-1)=\frac{abq^n(x_n-c)(x_n-1/c)(x_n-d)(x_n-1/d)}{(x_n-aq^n)(x_n-bq^n)},
\]
where $a,b,c,d$ and $q$ are parameters \cite{RGH}.
It can be rewritten as the following form:
\begin{equation}
\Phi_n: \left\{
\begin{array}{cl}
x_{n+1}&=\dfrac{1}{x_n}\left(\dfrac{abq^n(x_n-c)(x_n-1/c)(x_n-d)(x_n-1/d)}{(x_n-aq^n)(x_n-bq^n)(x_ny_n-1)}+1\right),\\
y_{n+1}&=x_n.
\end{array}
\right.
\label{qP5}
\end{equation}
\begin{prp}
Suppose that $a,b,c,d,q$ are in $\{1,2,\cdots, p-1\}$ and $a,b,c,d,c^{-1},d^{-1}$ are distinct from each other, then the mapping \eqref{qP5} has almost good reduction   
modulo $\mathfrak{p}$ on the domain 
$\mathcal{D}^{(n)}:=\{(x,y)\in \Z_p^2\ | x\neq aq^n,bq^n,\ xy\neq 1\}$.
\label{PropqP5}
\end{prp}

{\bf Proof}
The calculation is quite lengthy: we need large amount of memory space to finish the task.
For simplicity, we only treat the case of $n=0$.

(i) If $\tilde{x}_n=a$,
\[
\widetilde{\Phi_n^3(x_n,y_n)} = \widetilde{\Phi_n^3}(\tilde{x}_n=a,\tilde{y}_n)=\left(\frac{1}{bq},bq\right).
\]

(ii) If $\tilde{x}_n=b$,
\[
\widetilde{\Phi_n^3(x_n,y_n)} = \widetilde{\Phi_n^3}(\tilde{x}_n=b,\tilde{y}_n)=\left(\frac{1}{aq},aq\right).
\]

(iii) If $\tilde{x}_n\tilde{y}_n=1$,
\[
\widetilde{\Phi_n^3(x_n,y_n)} = \widetilde{\Phi_n^3}\left(\tilde{x}_n,\tilde{y}_n=\frac{1}{\tilde{x}_n}\right)=\left(\frac{1}{abq\tilde{y}_n},abq\tilde{y}_n\right).
\]

\smallskip

These propositions strengthen our idea that the almost good reduction can be one criterion for integrability of systems over finite fields.

\section{Relation to Diophantine integrability}
Lastly we discuss the relationship between the systems over finite fields and the algebraic entropy of systems.
Let $\phi$ be a rational difference map and let the degree of the map $\phi$ be $d>0$. We denote the degree of the iterates $\phi^n$ by $\deg (\phi^n)=d_n$.
The na\"{i}ve composition suggests $d_n=d^n$, however, common factors can sometimes be eliminated, lowering the degree of the iterates.
The algebraic entropy $E$ of $\phi$ is the following well-defined quantity \cite{BV}:
\[
E:=\lim_{n\to\infty}\frac{1}{n}\log d_n.
\]
We can see that the limit exists from Fekete's lemma and the relation $d_{m+n}\le d_m+d_n$ for $m,n\ge 1$. A non-constant rational map has $E\ge 0$.
We can postulate from a lot of numerical examples that the mapping $\phi$ is
integrable if and only if $E=0$, that is, $d_n$ has a polynomial growth.
(Or we could define the integrability of the map by this criterion.)

We can construct an arithmetic analogue of the algebraic entropy which has first been introduced by R. Halburd in \cite{Halburd}.
If we consider the map with rational numbers as coefficients, and choose initial values to be rational numbers, then we have $x_n\in \mathbb{Q}$ for all $n\in\mathbb{Z}_{>0}$. The arithmetic complexity of rational numbers can be expressed by the height function $H(x)$:
\[
H(x)=\max\{|u|,|v|\},
\]
where $x=\frac{u}{v}$ and $u$ and $v$ are integers without common factors. ($H(0)=0$.)
The map $\phi$ is said to be `Diophantine integrable' if and only if $\log H(x_n)$ grows as slowly as some polynomial.
Thus we can define the arithmetic analogue of algebraic entropy as
\[
\epsilon:=\lim_{n\to\infty} \frac{1}{n}\log\left(\log H(x_n)\right),
\]
if the limit exists. The quantity $\epsilon$ might be called the `Diophantine entropy' of the map. This quantity has been rigorously calculated by A. N. W. Hone for third order recurrence relations \cite{Hone1}. In this paper, he demonstrates that some class of maps are not Diophantine integrable even if they pass singularity confinement test and have Laurent property.
We have to note that the value $\epsilon$ may depend on the choice of initial conditions $x_0, x_1\cdots$.
However in many cases the Diophantine entropy $\epsilon$ should be the same for generic initial conditions as discussed in \cite{Hone2}.
In the case of original algebraic entropy $E$, we can rigorously obtain the recurrence relation for the sequence $d_n$ of degrees of rational functions. The Diophantine entropy $\epsilon$ can be obtained in a similar manner, although it is sometimes not easy to estimate the elimination of common factors $p$ between the numerator and the denominator. 

In this paper we present some numerical examples of the height growth of the discrete dynamical systems:

(i) The Hietarinta-Viallet equation 
\begin{equation}
x_{n+1}+x_{n-1}=x_n+\frac{a}{x_n^2}, \label{HV1}
\end{equation}
It is conjectured to have $\epsilon=\log\left(\frac{3+\sqrt{5}}{2}\right)$, which corresponds to the original algebraic entropy $E=\log\left(\frac{3+\sqrt{5}}{2}\right)$. The value $E$ was first numerically obtained in \cite{HV} and then proved in \cite{Takenawa1}.
We give one example which supports this conjecture for $\epsilon$.
Let us suppose that $p=3$, $(x_0,x_1)=(1,3)$ and that the parameter $a=1$ in the Hietarinta-Viallet equation \eqref{HV1}. Then
\begin{eqnarray*}
\{\mbox{Integer part of}\ \log_3 H(x_n)\}_{n=0}^\infty&=&\{0,1,2,7,19,50,132,347,911,2385,6245,\\
&& 16352,42811,112082,293434,768221,\cdots\cdots\}.
\end{eqnarray*}
We can see that
\[
\left. \frac{\log_3 H(x_{n+1})}{\log_3 H(x_n)}\right|_{n\to\infty}\sim 2.61804 \sim \frac{3+\sqrt{5}}{2}.
\]

(ii) In the case of the equation \eqref{discretemap}, it is conjectured that $\epsilon=\log 3>0$ for $\gamma=3$, while $\epsilon=0$ for $\gamma=1,2$ and $\log H(x_n)$ has a polynomial growth of second degree for generic initial conditions.
These conjectures are completely in line with the well-known results that \eqref{discretemap} is integrable for $\gamma=1,2$, non-integrable for $\gamma=3$, and that the degree of the iterates
grows quadratically in the integrable cases.
We give two numerical examples.
First let us suppose that $p=3$, $\gamma=3$, $(x_0,x_1)=(1,3)$, and that the parameter in the equation \eqref{discretemap} is $a=2$. Then,
\[
\{\mbox{Integer part of}\ \log_3 H(x_n)\}_{n=0}^\infty=\{0,1,3,8,26,79,236,711,2133,6400,19201,\cdots\}.
\]
Therefore we see that
\[
\frac{\log_3 H(x_{n+1})}{\log_3 H(x_n)} \sim 3,\ \ \epsilon\sim \log 3.
\]
On the other hand, if $\gamma=1$, we have
\[
\{\mbox{Integer part of}\ \log_3 H(x_n)\}_{n=0}^\infty=\{0,1,1,1,2,3,3,4,5,7,7,8,9,12,13,14,16,18,\cdots\}.
\]
Therefore we see that
\[
\frac{\log_3 H(x_{n+1})}{\log_3 H(x_n)} \sim 1,\ \ \epsilon\sim 0.
\]
The rate of growth of $\log_3 H(x_n)$ seems to be quadratic:
if we estimate $x_0,\cdots,x_{100}$ using a cubic polynomial, we obtain
\[
\log_3 H(x_n)\sim 0.120+0.185n+0.0454n^2+5\times 10^{-7} n^3,
\]
which indicates a quadratic growth.

Therefore, in the cases given above, the Diophantine entropy $\epsilon$ motivated by the Diophantine integrability is conjectured to be equivalent to the (original) algebraic entropy $E$.
This idea is essentially equivalent to studying the growth of the number of digits in the numerator (or the denominator) of $x_n\in\mathbb{Q}$ when expressed as a fraction of $p$-adic integers. Therefore the procedure can be seen as an analogue of the algebraic entropy of a system over a finite field $\mathbb{F}_p$.
Further investigation into the Diophantine entropy and into its relationship with other integrability criteria is what we wish to address in future works.

\section{Concluding remarks}
In this article we investigated the discrete Painlev\'{e} equations over the finite fields.
Our prescription was to define the system over the non-Archimedean field larger than the finite field $\mathbb{F}_p$, and then to reduce the results to $\mathbb{F}_p$.
In particular, we treated the systems over the field of $p$-adic numbers $\mathbb{Q}_p$.
We reviewed an integrability test (almost good reduction), which is an arithmetic analogue of the singularity confinement test.
We have proved that several types of the $q$-discrete Painlev\'{e} equations have this property.
Thanks to this property, the time evolution of the discrete dynamical system can be well-defined over a finite field.
Finally we have reviewed the Diophantine integrability criterion, and have presented some numerical experiments to demonstrate that it can be seen a $p$-adic analogue of the algebraic entropy.
One of the future problems is to investigate the soliton systems reduced modulo a prime. In particular, the discrete KdV equation, the discrete KP equation and a discrete version of the Schr\"{o}dinger equation over finite fields are under investigation, and the study of them is expected to present new aspects on the theory of ultra-discrete systems and that of cellular automata.  We also wish to solve the initial value problems for these equations by constructing analogues of the tools such as the tropical geometry over the field of $p$-adic numbers.

\section*{Acknowledgement}
The author wish to thank Professors J. Mada, T. Mase, T. Tokihiro and R. Willox for helpful discussions.
This work was partially supported by Grant-in-Aid for JSPS Fellows 24-1379.
%%%%%%%%%%%%%%%%%%%%%%%%%%%%%%%%%%%%%%%%%%%%%%%%%%

\end{document}